\documentclass[journal]{IEEEtran}
\ifCLASSINFOpdf
\else
\fi

\usepackage{etex}

\usepackage[all]{xy}
\usepackage{graphicx}

\usepackage{amsfonts}
\usepackage{stmaryrd}
\usepackage{caption}
\usepackage{epsfig}
\usepackage{epstopdf}
\usepackage{amsmath}

\usepackage{amsthm}

\usepackage{listings}

\usepackage[linesnumbered,ruled]{algorithm2e}
\SetKwRepeat{Do}{do}{while}%

\SetCommentSty{mycommfont}
\usepackage{amssymb}
\usepackage{textcomp}
\usepackage{bm}
\usepackage{pdfpages}
\usepackage{epstopdf}
\usepackage[T1]{fontenc}
\usepackage{multirow}
\usepackage{comment}
\usepackage{tikz}
\usetikzlibrary{automata,positioning,shapes,arrows}
\usepackage{subfig}

\usepackage{array}
\newcolumntype{P}[1]{>{\centering\arraybackslash}p{#1}}

\newtheorem{theorem}{Theorem}
\newtheoremstyle{exampstyle}
{\topsep} 
{\topsep} 
{} 
{} 
{\bfseries} 
{.} 
{.5em} 
{} 
\theoremstyle{exampstyle} \newtheorem{example}{Example}
\theoremstyle{exampstyle} 
\theoremstyle{exampstyle} 
\theoremstyle{exampstyle} \newtheorem{definition}{Definition}
\theoremstyle{exampstyle} \newtheorem{lemma}{Lemma}
\theoremstyle{exampstyle} 

{\tiny }

\begin{document}
%
\title{Counterexample-Guided Abstraction Refinement for POMDPs}
%
%
%

\author{Xiaobin~Zhang, 
        Bo~Wu, 
        and~Hai~Lin, 
\thanks{The authors were with Department of Electrical Engineering, University of Notre Dame, Notre Dame, IN 46556, USA.
	{\tt\small (xzhang11@nd.edu;~bwu3@nd.edu;~hlin1@nd.edu)}.}
}

%
%

\markboth{~}%
{Shell \MakeLowercase{\textit{et al.}}: Bare Demo of IEEEtran.cls for IEEE Journals}
%



\maketitle

\begin{abstract}
	Partially Observable Markov Decision Process (POMDP) is widely used to model probabilistic behavior for complex systems. Compared with MDPs, POMDP models a system more accurate but solving a POMDP generally takes exponential time in the size of its state space. This makes the formal verification and synthesis problems much more challenging for POMDPs, especially when multiple system components are involved. As a promising technique to reduce the verification complexity, the abstraction method tries to find an abstract system with a smaller state space but preserves enough properties for the verification purpose. While abstraction based verification has been explored extensively for MDPs, in this paper, we present the first result of POMDP abstraction and its refinement techniques. The main idea follows the counterexample-guided abstraction refinement (CEGAR) framework. Starting with a coarse guess for the POMDP abstraction, we iteratively use counterexamples from formal verification to refine the abstraction until the abstract system can be used to infer the verification result for the original POMDP. Our main contributions have two folds: 1) we propose a novel abstract system model for POMDP and a new simulation relation to capture the partial observability then prove the preservation on a fragment of Probabilistic Computation Tree Logic (PCTL); 2) to find a proper abstract system that can prove or disprove the satisfaction relation on the concrete POMDP, we develop a novel refinement algorithm. Our work leads to a sound and complete CEGAR framework for POMDP.
\end{abstract}


%
\IEEEpeerreviewmaketitle

\section{Introduction}
\label{intro}
Probabilistic behavior widely exists in practice for complex systems. To model such systems, Markov decision processes (MDPs) can be used to capture both nondeterminisms in system decision making and probabilistic behaviors in state transitions. As an extension of MDP, partially observable MDP (POMDP) considers another layer of uncertainties by assuming that the system states are not directly observable. This partial observability allows POMDPs to model systems more accurately but introduces more computational expenses when solving a POMDP\cite{ma2008modelling}. This is especially the case when the formal verification and synthesis problems are considered regarding system properties. 

Consider system properties that can be specified by temporal logics such as Probabilistic Computation Tree Logic (PCTL) \cite{rutten2004mathematical}. The model checking problem answers whether or not a system satisfies a given PCTL specification, for example, the probability to reach a set of bad states is less than 0.1\%. For MDPs, the model checking of PCTL formula can be solved in polynomial time in the size of state space \cite{baier2008principles}, while for POMDPs it generally requires exponential time \cite{madani1999undecidability,goldsmith1998complexity,chatterjee2016decidable}. When the system contains multiple sub-systems, the compositional verification can dramatically increase the computational cost for MDPs and can be even worse for POMDPs. This is known as the state space explosion problem.

As an important technique to conquer the state space explosion, abstraction is a method for abstracting the state space of the original (concrete) system and creating an abstract system that contains smaller state space but can still conservatively approximate the behavior of the concrete system \cite{hermanns2008probabilistic,chadha2010counterexample}. With a valid abstract system, the model checking problem can be reasoned on a smaller state space and its solution for the concrete system can be inferred correspondingly. While the abstraction methods have been proposed and applied to probabilistic systems such as MDPs in recent years  \cite{hermanns2008probabilistic,chadha2010counterexample,feng2010compositional,kwiatkowska2010assume,kwiatkowska2013compositional,he2016learning}, to the best of the authors' knowledge the abstraction methods have not been extensively explored for POMDPs where the partial observability brings new challenges. 

Motivated by the power of abstraction method in reducing model checking complexity, we consider here the abstraction method for POMDPs. To have the abstraction method well established for POMDPs, two key questions must be answered: 1) how to define the form of the abstract system; 2) how to find a proper abstract system. For non-probabilistic systems and MDPs, counterexample-guided abstraction refinement (CEGAR) framework is widely considered to answer the second question with different abstraction forms being defined \cite{clarke2000counterexample,hermanns2008probabilistic,chadha2010counterexample,komuravelli2012assume}. Generally, CEGAR starts with a very coarse abstraction and uses the counterexamples diagnosed from the abstract systems to iteratively refine the abstraction.     

In this article, we propose a sound and complete CEGAR framework for POMDPs, which allows the automatic reasoning on finding a proper abstraction for POMDPs. A safety fragment of PCTL (safe-PCTL)\cite{chadha2010counterexample} with finite horizon is considered as the system specification. As the abstraction for POMDPs, z-labeled 0/1-weighted automata ($0/1$-WA$^z$) are extended from $0/1$-weighted automata ($0/1$-WA) \cite{baier2008principles, he2016learning} by defining the observation labeling function for discrete states. Given a POMDP, its corresponding $0/1$-WA$^z$ can be defined with the observation labeling function representing the observation information. We then propose a simulation relation, \textit{safe simulation}, for $0/1$-WA$^z$ and prove that the safe simulation preserves the properties considered. This gives the foundation of the soundness of our CEGAR framework.

With $0/1$-WA$^z$ and safe simulation relation, we further address a novel CEGAR framework to find a proper $0/1$-WA$^z$ as the abstraction of POMDP. Initially, we start with the coarsest abstract system generated from a quotient construction and iteratively check the satisfaction relation of the given specification on the abstract system. Counterexamples from model checking on $0/1$-WA$^z$ are derived following algorithms in \cite{han2009counterexample} in the forms of a finite set of finite paths that violate the specification with large enough accumulative probability. Given these counterexamples, we verify whether or not these counterexamples are real witnesses for violation of the given specification on the concrete system. If not, we use these spurious counterexamples to refine the quotient construction and update the abstract system until satisfaction relation is proved to be true or real counterexample has been found for the concrete system.

\subsection{Our contributions}
The technical contributions are summarized in the order in which they appear in the article as follows:
\begin{itemize}
	\item We propose $0/1$-WA$^z$ as the form of abstract systems for POMDPs and give the conversion rules between $0/1$-WA$^z$ and POMDP. On $0/1$-WA$^z$, we design a new notion of simulation relation, safe simulation, which is extended from strong simulation relation for MDPs \cite{segala1995probabilistic}. With safe simulation relation capturing the property of the partial observability from POMDP, we prove the preservation of finite horizon safe-PCTL between the abstract system ($0/1$-WA$^z$) and the concrete system ($0/1$-WA$^z$ induced from the original POMDP).    
	
	\item We present a CEGAR framework to automatically find a proper abstract system to prove or disprove the satisfaction relation of the concrete system. We first define the quotient construction rules to generate candidate abstract systems based on a given state space partition. Starting with the coarsest partition, we iteratively refine the abstract systems using counterexamples returned from model checking on the abstract systems. Following our refinement algorithm, we prove the soundness and completeness of our CEGAR framework for POMDP. An example is given to show the effectiveness of our CEGAR framework to reduce the state space size of POMDP.
\end{itemize}

\subsection{Related work}

Abstraction methods and the refinement algorithms based on counterexamples have been extensively explored for probabilistic systems. In \cite{hermanns2008probabilistic}, probabilistic automata (PA) is considered. Compared with MDPs the transition distribution under an action in PA is not deterministic, i.e., there can be multi distributions defined for single state and action pair. With quotient construction to build abstract system, strong simulation relation \cite{segala1995probabilistic} is used to describe the properties shared between abstract and concrete system since strong simulation preserves the safety fragment of PCTL specification. Strong simulation relation also inspires us to derive safe simulation relation for POMDPs. To refine the abstract, CEGAR based on the interpolation method \cite{Henzinger:2004:AP:982962.964021} is applied with counterexamples being returned from model checking following \cite{han2007counterexamples}.  

Still using strong simulation relation to restrict abstraction, \cite{chadha2010counterexample} considers CEGAR on a special class of MDPs. Instead of an action set, the transition options for each state are determined by a set of distributions and the adversary will select a distribution to transit. From this point of view, each distribution is defined under a unique action. With system specification given by PCTL, different counterexample forms for MDPs are discussed. For safety PCTL, an MDP and simulation relation pair is used as the counterexample, then the spuriousness of the counterexample can be checked by standard simulation checking algorithm \cite{baier2000deciding}. Based on that a refinement will try to split the invalidating state. 

While quotient construction based on state space partition can reduce the size of state space, recent work in \cite{he2016learning} uses $0/1$-WA as the assumptions for MDPs and reduces the transition links and terminal nodes in multi-terminal binary decision diagrams (MTBDDs) representation of MDPs. The $L^*$ algorithm is then applied to learn an MTBDD as a proper assumption. But their method limits the choice of transition probabilities for $0/1$-WA to be either the transition probabilities from MDP or 1/0, which limits the form of abstractions.

Beside abstraction based assumption for the compositional verification of MDPs, there are also results on classical assumptions \cite{kwiatkowska2010assume,feng2010compositional,feng2011learning}. Compared with the abstraction based assumptions mentioned earlier, classical assumptions usually preserve only linear-time properties instead of branching-time properties between the assumption and the original system. This kind of assumptions can lead to a smaller state space but may not preserve rich enough properties and hard to guarantee the completeness.

For CEGAR framework of POMDP, our work is mainly inspired from \cite{hermanns2008probabilistic,chadha2010counterexample} where the safe-PCTL specification is considered and simulation based abstraction is discussed. 

\subsection{Outline of the article}
The rest of this article is organized as follows. In Section 2, we give necessary definition and notations. The $0/1$-WA$^z$ as the abstract system for POMDP and safe simulation relation are introduced in Section 3. Then the CEGAR framework is presented in Section 4 and an illustration example is given in Section 5. Finally, we conclude this article with future work in Section 6.

\section{Preliminaries}
\subsection{Probabilistic system models}

\begin{definition}
	\cite{rutten2004mathematical} An MDP is a tuple $\mathcal{M}=(S,\bar{s},A,T)$ where
	\begin{itemize}
		\item $S$ is a finite set of states;
		\item $\bar{s}\in S$ is the initial state;
		\item $A$ is a finite set of actions;
		\item $T:S\times A \times S\rightarrow [0,1]$ is a transition function
	\end{itemize}
\end{definition}
\noindent Here the probability of making a transition from state $s\in S$ to state $s'\in S$ under action $a\in A$ is given by $T(s,a,s')$. For MDP, it is required that $\sum_{s'\in S}T(s,a,s')=1$, $\forall s\in S, a\in A$. Under this requirement, the terminating states can be modeled by adding a self-loop with the probability $1$ \cite{rutten2004mathematical}. As a special case, DTMC is an MDP with only one action defined for each state $s\in S$ and the transition function can be reduced as $T:S\times S\rightarrow [0,1]$. To analyze the behavior of MDP with additional information, we can define a labeling function $L: S\rightarrow 2^{AP}$ that assigns each state $s\in S$ with a subset of atomic propositions $AP$. 

For MDPs, the states are assumed to be fully observable. As a generalization, POMDPs consider system states with partial observability. 
\begin{definition}
	A POMDP is a tuple $\mathcal{P}=\{\mathcal{M},Z,O\}$ where
	\begin{itemize}
		\item $\mathcal{M}$ is an MDP;
		\item $Z$ is a finite set of observations;		
		\item $O:S\times Z\rightarrow [0,1]$ is an observation function.
	\end{itemize}
\end{definition}
\noindent Instead of direct observability, POMDP assigns a probability distribution over the observation set $Z$ for every state $s\in S$, which is described by the observation function $O$. Here the probability of observing $z\in Z$ at state $s\in S$ is given by $O(s,z)$.

In this paper, we consider finite POMDPs, in which $S,~A,~Z$ and $AP$ are finite sets.

\subsection{Paths and adversaries}
In the MDP $\mathcal{M}=(S,\bar{s},A,T)$, a $path$ $\rho$ is a non-empty sequence of states and actions in the form 
\begin{align*}
\rho = s_0a_0s_1a_1s_2\ldots,
\end{align*}
where $s_0 = \bar{s}, ~s_i\in S,~a_i\in A$ and $T(s_i,a_i,s_{i+1})\geq0$ for all $i\geq0$ \cite{rutten2004mathematical}. Let $\rho(i)$ denote the $i$th state $s_i$ of a path $\rho$, $\rho^-_i$ denote the prefix ending in the $i$th state $s_i$ and $\rho^+_i$ denote the suffix starting from the $i$th state $s_i$. Let $|\rho|$ denote the length of $\rho$ which is the number of transitions. The set of paths in $\mathcal{M}$ is denoted as $Path_\mathcal{M}$ and its set of corresponding prefixes is denoted as $Pref_{\mathcal{M}}$. Let $Pr^{\rho}(i,j)$ ($i\leq j$) stand for the product of the transition probabilities from the $i$th state to the $j$th state ($i<j$) on path $\rho$. Thus   
\begin{align*}
Pr^{\rho}(i,j)=\prod^{j-1}_{k=i}T(s_k,a_{k},s_{k+1}),~i<j,
\end{align*}
and $Pr^{\rho}(i,j)=1$ if $i=j$. 

To synthesize about an MDP, an $adversary$ resolves the system nondeterminism given system executions. A pure adversary is a function $\sigma:Pref_{\mathcal{M}}\rightarrow A$, which maps every finite path of $\mathcal{M}$ onto an action in $A$. A randomized adversary is a function $\sigma:Pref_{\mathcal{M}}\rightarrow Dist(A)$, which maps every finite paths of $\mathcal{M}$ onto a distribution over action set $A$. Under an adversary function $\sigma$, the set of possible system paths is denoted as $Path^\sigma_\mathcal{M}$. While the pure adversary is a special case of the randomized adversary, in this paper we consider pure adversaries and we will discuss in the next subsection that pure adversaries are as powerful as randomized adversaries for the type of system properties we consider. 

Compared with MDPs, POMDPs do not have direct observability over the states. For a POMDP $\mathcal{P}=\{S,\bar{s},A,Z,T,O\}$, define the observation sequence of a path $\rho =s_0a_0s_1a_1s_2\ldots$ as a unique sequence $obs(\rho)=z_0a_0z_1a_1z_2\ldots$ where $z_i\in Z$ and $O(s_i,z_i)>0$ for all $i\geq0$. The observation sequence is also known as the history information. To synthesize about POMDP, the adversary is required to be \textit{observation-based}, where for all $\rho,\rho'\in Pref_{\mathcal{M}}$, $\sigma(\rho)=\sigma(\rho')$ if $obs(\rho)=obs(\rho')$.  

\subsection{Probabilistic computation tree logic}

To represent and synthesize the design requirements or control objectives, Probabilistic Computation Tree Logic (PCTL) \cite{rutten2004mathematical} is considered in this paper. As the probabilistic extension of the Computation Tree Logic (CTL) \cite{clarke1982design}, PCTL adds the probabilistic operator $P$, which is a quantitative extension of CTL's $A$ (always) and $E$ (exist) operators \cite{baier2008principles,kwiatkowska2007stochastic}.  

\begin{definition} \cite{rutten2004mathematical}
	The syntax of PCTL is defined as
	\begin{itemize}
		\item State formula $\phi::=true~|~\alpha~|~\neg \phi~|~\phi\wedge\phi~|P_{\bowtie p}[\psi]$,
		\item Path formula $\psi::=X\phi~|~\phi~\mathcal{U}^{\leq k}\phi~|~\phi~\mathcal{U}~\phi,$
	\end{itemize}
	where $\alpha\in AP$, $\bowtie\in\{\leq,<,\geq,>\}$, $p\in[0,1]$ and $k\in\mathbb{N}$.
\end{definition}
\noindent Here $\neg$ stands for "negation", $\wedge$ for "conjunction", $X$ for "next", $\mathcal{U}^{\leq k}$ for "bounded until" and $\mathcal{U}$ for "until". Specially, $P_{\bowtie p}[\psi]$ takes a path formula $\psi$ as its parameter and describes the probabilistic constraint. Note that a PCTL formula is always a state formula and path formulas only occur in $\mathcal{P}$ operator.

\begin{definition}\cite{rutten2004mathematical,zhang2015learning}
	For an labeled MDP $\mathcal{M}=(S,\bar{s},A,T,L)$, the satisfaction relation $\vDash$ for any states $s\in S$ is defined inductively as follows
	\begin{align*}
	s &\vDash true,~\forall s\in S;\\
	s & \vDash \alpha \Leftrightarrow \alpha\in L(s);\\
	s &\vDash \neg\phi \Leftrightarrow s\nvDash\phi;\\
	s &\vDash \phi_1\wedge\phi_2 \Leftrightarrow s\vDash\phi_1\wedge s\vDash\phi_2;\\
	s &\vDash P_{\bowtie p}[\psi] \Leftrightarrow Pr(\{\rho\in Path^{\sigma}_\mathcal{M}|~ \rho \vDash\psi \})\bowtie p,~\forall \sigma\in \Sigma_{\mathcal{M}},
	\end{align*}
	where $\Sigma_\mathcal{M}$ is the set of all adversaries and for any path $\rho\in Path_\mathcal{M}$
	\begin{align*}
	\rho &\vDash X\phi \Leftrightarrow \rho(1)\vDash \phi;\\
	\rho &\vDash \phi_1~\mathcal{U}^{\leq k}\phi_2 \Leftrightarrow \exists i\leq k, \rho(i)\vDash \phi_2\wedge\rho(j)\vDash\phi_1,\forall j<i;\\
	\rho &\vDash \phi_1~\mathcal{U}\phi_2 \Leftrightarrow \exists k\geq 0, \rho\vDash\phi_1~\mathcal{U}^{\leq k}\phi_2.
	\end{align*}
\end{definition}

Compared with MDPs, the adversaries for POMDPs must be observation-based. Therefore the PCTL satisfaction relation for $s \vDash P_{\bowtie p}[\psi]$ is limited to consider observation-based adversaries for POMDPs, while other satisfaction relations can be inherited from MDP cases. 

In this paper, we consider the safety fragment of PCTL (safe-PCTL) \cite{chadha2010counterexample} which is given in conjunction with the liveness fragment as follows.
\begin{itemize}
	\item $\phi_S:=true~|~\alpha~|~\phi_S\wedge\phi_S~|~\phi_S\vee\phi_S~|P_{\unlhd p}[ \phi_L~\mathcal{U}~\phi_L]$,
	\item $\phi_L:=true~|~\alpha~|~\phi_L\wedge\phi_L~|~\phi_L\vee\phi_L~|\neg P_{\unlhd p}[ \phi_L~\mathcal{U}~\phi_L]$,
\end{itemize}
where $\alpha\in AP$, $\unlhd\in\{\leq,<\}$, $p\in[0,1]$. If we restrict $\unlhd$ to be $\leq$ in the above grammar, we will get \emph{strict liveness} and \emph{weak safety} fragments of PCTL. While reasoning some logic properties for POMDPs is undecidable \cite{chatterjee2013survey}, we focus on safe-PCTL with finite-horizon which is decidable. 
And without losing generality, in the rest part of the article, we will mainly consider the safe-PCTL for bounded until $\phi=P_{\unlhd p}[\phi_1~\mathcal{U}^{\leq k}\phi_2]$ as the type of specifications to illustrate our framework. 

Generally, a randomized adversary is more powerful than a pure adversary, its special case. But for the PCTL fragment considered in this paper, the pure adversaries and randomized adversaries have the same power for MDPs and POMDPs in the sense that restricting the set of adversaries to pure strategies will not change the satisfaction relation of the PCTL fragment \cite{chatterjee2010randomness}. The intuitive justification of this claim is that if we are just interested in upper and lower bounds to the probability of some events to happen, any probabilistic combination of these events stays within the bounds. Moreover, deterministic adversaries	are sufficient to achieve the bounds \cite{chatterjee2010randomness}.

\subsection{Model checking and counterexample selection}
For MDPs, the probabilistic model checking of PCTL specification has been extensively studied \cite{rutten2004mathematical}. Depending on whether $\bowtie$ in the specification gives upper or lower bound, PCTL model checking of MDPs solves an optimization problem by computing either the minimum or maximum probability over all adversaries \cite{rutten2004mathematical}. Due to its full observability, MDP model checking can be solved generally in polynomial time in the size of the state space and the computational complexity has been discussed extensively in \cite{baier2008principles}. There also exist model checking software tools for MDPs, such as PRISM \cite{kwiatkowska2011prism} and storm \cite{dehnert2016probabilistic}.

As a generalization of MDPs to consider uncertainties in both transitions and observations, POMDPs model system more accurate but the model checking of POMDPs is more expensive \cite{ma2008modelling}. Generally, the solution for POMDP model checking is in exponential time in the size of the state space or even undecidable for lots of properties \cite{madani1999undecidability,goldsmith1998complexity,chatterjee2016decidable}. To solve the POMDP model checking problem for specification $\phi=P_{\unlhd p}[\phi_1~\mathcal{U}^{\leq k}\phi_2]$ considered in this article, one possible approach is modifying the transition structure of POMDP to make states $s\models\neg\phi_1$ and states $s\models\phi_2$ absorbing, and designing the reward scheme that assigns 0 to intermediate transitions and 1 to the final transitions on $s\models\phi_2$ when depth $k$ is reached \cite{sharan2014formal}. Then the model checking problems can be formulated as a classic POMDP optimization problem that can be solved by, for example, value iteration method \cite{zhang2001speeding,pineau2006anytime}. 

While model checking answers whether or not a given specification can be satisfied, an adversary will be returned as the witness if the specification is violated. As one step further, counterexample selection gives a particular system path or set of paths in the system that violate the specification as a detailed evidence. Such a path or set of paths is called counterexample. For non-probabilistic systems, finding counterexamples is done by finding a path that violates the specification. However, for the probabilistic systems, the probabilities of the paths also need to be considered. Given an adversary, the nondeterminism of MDPs and POMDPs can be solved which give the induced DTMCs. Then the counterexample selection for MDPs and POMDPs can be considered as for DTMCs since model checking problem will generate the witness adversary if the specification is violated. For DTMCs, in \cite{han2009counterexample} the counterexamples are considered to be a finite set of executions and finding the counterexamples that carry enough evidence is equivalent to solve shortest path problem on the graph generated from DTMCs. But for some properties, e.g., $\mathcal{P}_{<1}(true~\mathcal{U}~\phi)$, there may not exist a finite set of paths that witnesses the violation as shown in \cite{chadha2010counterexample}. While there are different counterexample forms proposed to reason different properties and cases, readers may refer to \cite{chadha2010counterexample} for a comprehensive summary. 

For the finite horizon safe-PCTL considered in this paper, counterexample can always be represented by a finite set of finite paths. In the case of $\leq$, this is clear since it belongs to the weak safety fragment of PCTL, which is a subset of properties considered in \cite{han2009counterexample}; in the case of $<$, the set of paths as the counterexample is finite since the specification is only on the finite horizon for a finite system. Therefore in this paper we apply the framework and algorithms from \cite{han2009counterexample} and COMICS tool \cite{jansen2012comics} is used as the software package to generate counterexamples.

\section{Abstraction and simulation relation}
\label{sec3}
For CEGAR in MDPs, different abstract system forms are proposed \cite{hermanns2008probabilistic,chadha2010counterexample,komuravelli2012assume} depending on the verification goals and considered specifications. For POMDP with partial observability, we present $0/1$-WA$^z$ as the abstract system form, which is extended from $0/1$-WA considered in \cite{he2016learning}. Based on $0/1$-WA$^z$, we further propose a simulation relation --- safe simulation that can preserve the finite horizon safe-PCTL specification between the abstract $0/1$-WA$^z$ system and the concrete POMDP system. These give the foundation of our sound and complete CEGAR framework for POMDP.

\subsection{z-labeled 0/1-weighted automata}

\begin{definition}\label{Def_zWA}
	A z-labeled 0/1-weighted automaton $0/1$-WA$^z$ $\mathcal{M}^z$ is a tuple $\{S,\bar{s},\\A,T,Z,L,L^z\}$,
	\begin{itemize}
		\item $S$ is a finite set of states;
		\item $\bar{s}$ is the initial state;
		\item $A$ is a finite set of actions;		
		\item $T:S\times A \times S \rightarrow [0,1]$ is a transition function;
		\item $Z$ is a finite set of observation labels;		
		\item $L:S\rightarrow 2^{AP}$ is a atomic proposition labeling function with $AP$ being a finite set of atomic propositions;
		\item $L^z:S\rightarrow Z$ is a z-labeling function.
	\end{itemize}
\end{definition}

With $L^z$ describing the observable information for each state, the system execution of $0/1$-WA$^z$ behaves in the same way of POMDP after we embed the transition and observation function in POMDP into transition function in $0/1$-WA$^z$. From this point of view, $0/1$-WA$^z$ can naturally be an abstraction form for POMDP. Due to the partial observability of states in $0/1$-WA$^z$, its adversaries must be also observation-based and the model checking for $0/1$-WA$^z$ can be inherited from POMDP model checking, which is further discussed in the appendix section. 

\noindent  \textbf{Remark}: Compared with $0/1$-WA model described in \cite{he2016learning}, $0/1$-WA$^z$ adds an observation set $Z$ and a z-labeling function $L^z$ to represent the observable information on each state in $S$. But the definition of cylinder set and $\sigma$ algebra in $0/1$-WA$^z$ can still be straightforwardly inherited from $0/1$-WA. 

Given a POMDP, we can define its corresponding $0/1$-WA$^z$, which can be seen as adding observation labels directly to its guided MDP proposed in our previous work \cite{zhang2015learning,zhang2016assume}.

\begin{definition}\label{Def_MDP+}
	\cite{zhang2015learning}	Given a POMDP $\mathcal{P}=\{S,\bar{s},A,Z,T,O\}$, its guided MDP $\mathcal{M}^{+}$ is a tuple $\{X^+,\bar{x}^+,A,T^+\}$,
	\begin{itemize}
		\item $X^+=\{x^+|x^+=[s,z],s\in S,~z\in Z\}\cup\{\bar{s}\}$ is a finite set of states;
		\item $\bar{x}^+=\bar{s}$ is the initial state;
		\item $A$ is a finite set of actions;		
		\item $T^+(\bar{s},a,[s',z']):=T(\bar{s},a,s')\cdot O(s',z')$;
		\item $T^+([s,z],a,[s',z']):=T(s,a,s')\cdot O(s',z')$.
	\end{itemize}
\end{definition}
Given a POMDP $\mathcal{P}=\{S,\bar{s},A,Z,T,O\}$, we define the corresponding $0/1$-WA$^z$ $\mathcal{M}^z:\{\mathcal{M}^+,Z^z,L^z\}$  where 
\begin{itemize}
	\item $\mathcal{M}^+$ is the guided MDP for $\mathcal{P}$;
	\item $Z^z=Z\bigcup\{init\}$; 
	\item $L^z(\bar{s})=init$ and $L^z([s,z])=z, s\in S, z\in Z$.
\end{itemize}
\noindent Here the special observation $init$ is used to label the initial state $\bar{s}$. While we assume the initial state of POMDP is known, if the initial condition given for the POMDP is a distribution over $S$, a dummy state can be added as the new initial state with transition links to all other $s\in S$ at the probability specified by the initial distribution \cite{zhang2015learning}.

\begin{example}\label{example1}
	Consider a POMDP $\mathcal{P}=\{S,\bar{s},A,Z,T,O\}$\footnote{This POMDP model is extended from the MDP model considered in \cite{chadha2010counterexample}}, where
	\begin{itemize}
		\item $S$ consists of $n+3$ states ($n$ is an even number and $n\geq 2$). $S=\{s_f,s_{even},s_{odd}\}\cup \{s_0,s_1,...,s_{n-1}\}$;
		\item $\bar{s} = s_0$;	
		\item $A=\{a,b\}$;
		\item $Z=\{z_f, z_{even}, z_{odd}\}$.
	\end{itemize}
	
	\begin{figure}
		\centering
		
		\begin{tikzpicture}[shorten >=1pt,node distance=3.5cm,on grid,auto, bend angle=20, thick,scale=0.7, every node/.style={transform shape}] 
		\node[state,initial] (q0) {$s_0$};
		\node[state] (q1) [right=of q0] {$s_1$};
		\node        (qd) [right=of q1] {$\cdots$}; 
		\node[state] (qn) [right=of qd] {$s_{n-1}$};    		
		\node[state] (qeven) [above=of q1] {$s_{even}$};
		
		\node[state] (qodd) [below=of q1] {$s_{odd}$};
		\node[state, fill=orange] (qf) [right=of qodd] {$s_f$};
		
		\path[->]
		(q0) edge node {$a:0.5$} (q1)  
		edge [out=245, in=240]  node [pos=0.3, sloped, below] {$a:0.25$} (qf)
		edge node [pos=0.5, sloped, above] {$a:0.25$} (qeven)
		edge [loop above, pos = 0.5,sloped, above] node {$b:1$} ()
		(q1) edge  node {$a:0.5$} (qd) 
		edge [loop above, pos = 0.5,sloped, above] node {$b:1$} ()
		edge node[pos=0.5, sloped, above] {$a:0.25$} (qf)
		edge node [pos=0.5, sloped, above]{$a:0.25$} (qodd)
		(qd) edge node{$a:0.5$} (qn)
		(qodd) edge [loop left, pos = 0.7,sloped, above] node {$a:0.5$} () 
		edge node[pos=0.5, sloped, below] {$a:0.5$} (qf)
		(qn) edge  node [pos=0.5, sloped, above]{$a:0.25$} (qf) 
		edge  node [pos=0.5, sloped, above]{$a:0.25$} (qodd) 
		edge [loop above, pos = 0.5,sloped, above] node {$a:0.5,b:1$} ()
		; 
		\end{tikzpicture} 
		
		\caption{The POMDP model $\mathcal{P}$ in Example \ref{example1}}
		\label{fig:P}
	\end{figure}
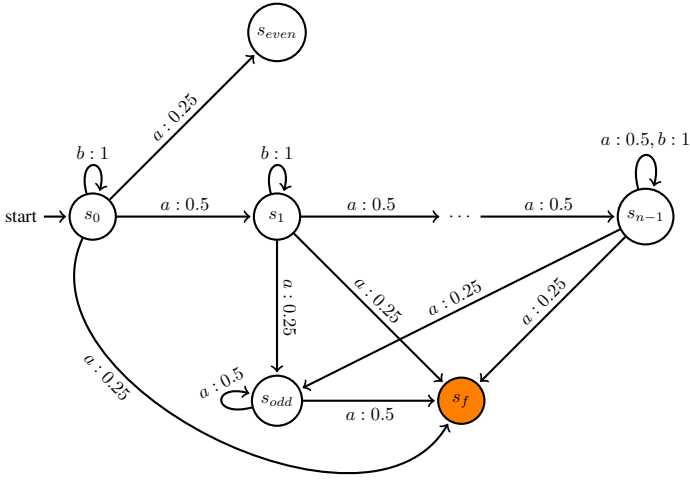
	
	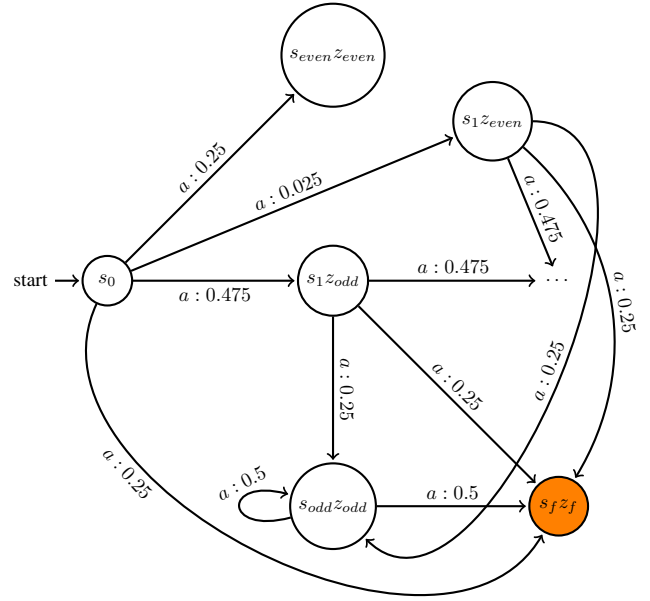
\begin{figure}
		\centering
		\begin{tikzpicture}[shorten >=1pt,node distance=4cm,on grid,auto, bend angle=40, thick,scale=0.75, every node/.style={transform shape}] 
		\node[state,initial] (q0) {$s_0$};
		\node[state] (q1) [right=of q0] {$s_1z_{odd}$};
		\node[state] (q2) [above right=of q1] {$s_1z_{even}$};
		\node        (qd) [right=of q1] {$\cdots$}; 
		\node[state] (qeven) [above=of q1] {$s_{even}z_{even}$};
		\node[state] (qodd) [below=of q1] {$s_{odd}z_{odd}$};
		\node[state, fill=orange] (qf) [right=of qodd] {$s_fz_f$};					
		\path[->]
		(q0) edge node [pos=0.5, sloped, below]{$a:0.475$} (q1)  
		edge node [pos=0.5, sloped, above]{$a:0.025$} (q2)
		edge [out=245, in=240]  node [pos=0.3, sloped, below] {$a:0.25$} (qf)
		edge node [pos=0.5, sloped, above] {$a:0.25$} (qeven)
		(q1) edge  node {$a:0.475$} (qd) 
		edge node[pos=0.5, sloped, above] {$a:0.25$} (qf)
		edge node [pos=0.5, sloped, above]{$a:0.25$} (qodd)
		(q2) edge  node [pos=0.5, sloped, above] {$a:0.475$} (qd) 
		edge [bend left] node[pos=0.5, sloped, above] {$a:0.25$} (qf)
		edge [out=0, in=-45] node [pos=0.5, sloped, above]{$a:0.25$} (qodd)
		(qodd) edge [loop left, pos = 0.7,sloped, above] node {$a:0.5$} () 
		edge node[pos=0.5, sloped, above] {$a:0.5$} (qf)				 
		; 
		
		\end{tikzpicture} 
		\caption{The corresponding $0/1$-WA$^z$ $\mathcal{M}^z$ for $\mathcal{P}$ in Example \ref{example1}}
		\label{fig:Mz}
	\end{figure}
	
	\begin{table}[h]
		\centering
		\caption{Observation matrix}
		\label{table:O}
		\begin{tabular}{c|ccc}
			$O(s,z)$	& $z_f$ & $z_{even}$ & $z_{odd}$ \\ \hline
			$s_f$         & 1     &            &           \\
			$s_{even}$    &       & 1          &           \\
			$s_{odd}$     &       &            & 1         \\
			$s_{i}$: even &       & 0.95       & 0.05      \\
			$s_{i}$: odd  &       & 0.05       & 0.95     
		\end{tabular}
	\end{table}

	The transition probability under both actions are given in Fig. \ref{fig:P}, where 
	\begin{align*}
	& T(s_i,b,s_i)=1,~i\in[0,n-1],\\
	& T(s_i,a,s_{i+1})=0.5,~i\in[0,n-2],\\
	& T(s_{n-1},a,s_{n-1})=0.5,\\
	& T(s_i,a,s_f)=0.25,~i\in[0,n-1],\\
	& T(s_i,a,s_{even})=0.25, \text{ for } i\text{: even and }i\in[0,n-1],\\
	& T(s_i,a,s_{odd})=0.25, \text{ for } i\text{: odd and }i\in[0,n-1],\\	
	& T(s_{odd},a,s_f)=0.5,\\
	& T(s_{odd},a,s_{odd})=0.5.
	\end{align*}
	The observation matrix is shown as follows in Table \ref{table:O}. Among $S$, the state $s_f$ represents a failure state with label $fail$ and is colored by orange in Fig. \ref{fig:P}. No other states has proposition labels. 
	
	For POMDP $\mathcal{P}$, the corresponding $0/1$-WA$^z$ $\mathcal{M}^z$ is shown in Fig. \ref{fig:Mz}.	
\end{example}

Based on the definition of $\mathcal{M}^z$, the paths from POMDP $\mathcal{P}$ and paths from $\mathcal{M}^z$ build a one-to-one mapping under a same observation-based adversary. This makes POMDP and its corresponding $0/1$-WA$^z$ equivalent and convertible for formal verification purposes. Notice that the parallel production of POMDP's state space and observation space during the generation of $0/1$-WA$^z$ $\mathcal{M}^z$ is mainly for the usage of theoretical proof of our CEGAR framework. Though the cross production generates a larger state space in $0/1$-WA$^z$, the random processes for state transitions and observation selections in POMDP have been encoded in the state transitions in $0/1$-WA$^z$. Therefore it will not introduce lots of computational expenses in our refinement algorithm introduced later.

\subsection{Simulation relation for $0/1$-WA$^z$}

With $0/1$-WA$^z$ as the possible form of the abstract systems for POMDP, we need to define the property preservation relation between the abstract system and the concrete system. For MDPs, simulation relations are widely used and discussed \cite{segala1995probabilistic, chadha2010counterexample}. Based on the strong simulation relation \cite{segala1995probabilistic} that preserves safe-PCTL for MDPs, we propose the safe simulation relation for $0/1$-WA$^z$ to preserve finite horizon safe-PCTL specifications. 

Consider two $0/1$-WA$^z$ $\mathcal{M}^z_1=\{S_1,\bar{s}_1,A,T_1,Z,L_1,L^z_1\}$ and $\mathcal{M}^z_2=\{S_2,\bar{s}_2,A,\\T_{2},Z,L_2,L^z_2\}$. For $a\in A$, $s_1\in S_1$ and $s_2\in S_2$, let $s_1\xrightarrow{a}\mu_1$ and $s_2\xrightarrow{a}\mu_2$ with $\mu_i(s')=T_i(s_i,a,s'),~s'\in S_i,~i\in\{1,2\}$. Let $Supp(\mu):=\{s|~\mu(s)>0\}$. Note that we still call $\mu_i$ a probability distribution for convenience but it is non-stochastic since its total probability mass can be larger than 1. Let $R\subseteq S_1\times S_2$ denote a binary relation between the state spaces of $\mathcal{M}^z_1$ and $\mathcal{M}^z_2$.
\begin{definition}\label{def:mu_sim}
	$\mu_1\sqsubseteq_{R}\mu_2$ if and only if there is a weight function $w:S_1\times S_2\rightarrow[0,1]$ such that\\
	1. $\mu_1(s_1)=\sum_{s_2\in S_2}w(s_1,s_2)$, $\forall$ $s_1\in S_1$,\\
	2. $\mu_2(s_2)\geq\sum_{s_1\in S_1}w(s_1,s_2)$,  $\forall$ $s_2\in S_2$,\\
	3. $w(s_1,s_2)=0$ if $L^z_1(s_1)\neq L^z_2(s_2)$ or $L_1(s_1)\neq L_2(s_2)$,\\
	4. $w(s_1,s_2)>0$ implies $s_1R s_2$, $\forall$ $s_1\in S_1$, $s_2\in S_2$.
\end{definition}	

\begin{definition}\label{def:safesim}
	$R$ is a safe simulation relation between two $0/1$-WA$^z$ $\mathcal{M}^z_1$ and $\mathcal{M}^z_2$ if and only if for every $s_1Rs_2$ and $s_1\xrightarrow{a}\mu_1$, there exists a $\mu_2$ with $s_2\xrightarrow{a}\mu_2$ and $\mu_1\sqsubseteq_{R}\mu_2$. 
	
	For $s_1\in S_1$ and $s_2\in S_2$, $s_2$ safely simulates $s_1$, denoted $s_1\preceq s_2$, if and only if there exists a safe simulation $T$ such that $s_1Ts_2$. 
	
	$\mathcal{M}^z_2$ safely simulated $\mathcal{M}^z_1$, also denoted $\mathcal{M}^z_1\preceq\mathcal{M}^z_2$, if and only if $\bar{s}_1\preceq\bar{s}_2$.
\end{definition}

\begin{lemma}\label{lemma:mu}
	$\mu_1\sqsubseteq_{R}\mu_2$ implies, $\mu_1(S)\leq\mu_2(R(S))$, for every $S\subseteq Supp(\mu_1)$.
\end{lemma}
\begin{proof}
	Let $w$ be the associated weight function, then 
	\begin{align*}
	\mu_1(S) & = \sum_{s\in S}\sum_{t\in R(S)}w(s,t) \\
	& =\sum_{t\in R(S)}\sum_{s\in S}w(s,t) \leq \mu_2(R(S)).
	\end{align*}

\end{proof}

Lemma \ref{lemma:mu} states that the safe simulation relation for probability distributions enlarges the probability masses over the support. In the next theorem, we prove that the safe simulation preserves of the finite horizon safe-PCTL specifications, where its enlargement effect for the sets of paths is shown.    

\begin{theorem}\label{theorem:sim}
	Let $\phi$ be a safe-PCTL specification with finite horizon. For a POMDP $\mathcal{P}$ and its corresponding $0/1$-WA$^z$ $\mathcal{M}^z_1$,  if $\mathcal{M}^z_1\preceq\mathcal{M}^z_2$ with $\mathcal{M}^z_2$ being $0/1$-WA$^z$, then $\mathcal{M}^z_2 \models \phi $ implies $\mathcal{P}\models \phi$.
\end{theorem}

\begin{proof}
	For finite horizon safe-PCTL specifications without probabilistic operator $P$, this theorem holds trivially since the initial states of $\mathcal{M}^z_2$ and $\mathcal{P}$ have the same atomic proposition labels.  
	
	Consider the specifications containing probabilistic operator with constraint, $\phi=P_{\unlhd p}[~]$. Assume $\mathcal{M}^z_2 \models \phi $ but $\mathcal{P} \not\models \phi$. For $\mathcal{M}^z_1$, there must exist an observation-based adversary, which generates a witness DTMC $\mathcal{M}_1$ that violates the specification $\phi$. We want to show, under the same adversary, the corresponding weighted DTMC $\mathcal{M}_2$ for $\mathcal{M}^z_2$ also violates $\phi$. Thus $\mathcal{M}_2$ can serve as the witness of $\mathcal{M}^z_2$'s violation and show $\mathcal{M}^z_2\not\models \phi$. Then by contradiction the theorem can be proved. 
	
	Let $R$ be the safe simulation relation between $\mathcal{M}^z_1$ and $\mathcal{M}^z_2$. Since $\mathcal{M}^z_1\not\models\phi$, there exists a finite set of finite paths that satisfy $\phi$ and their accumulated probability masses $\ntrianglelefteq p$.
	
	To show the contradiction, we first prove that for a finite set of paths with length $n$, $Path_s=\{\rho_s|\rho_s=s_0...s_n\}$ in $\mathcal{M}^z_1$, the corresponding set of paths $Path_t=\{\rho_t|\rho_t=t_0...t_n, \text{~s.t. }\exists \rho_s=s_0...s_n,\in Path_s,~s_i Rt_i\}$ in $\mathcal{M}^z_2$ satisfies $Pr(Path_t)\geq Pr(Path_s)$. For $i=0,\ldots,n$, let $s_i\in S(i)$ where $S(i)=\{s|s=\rho_s(i),\rho_s\in Path_s\}$. For paths in $Path_s$, let $\mu^s_i$ denote the (non-stochastic) probability distribution over $S(i)$ at step $i$. Correspondingly we could define $T_i$, $\mu^t_i$ for $Path_t$. By the definition of $Path_t$, we have $T(i) = R(S(i)))$. Clearly, $Pr(Path_s)=\mu^s_n(S(n))$ and $Pr(Path_t)=\mu^t_n(T(n))$. Now we are proving $\mu^s_j\sqsubseteq_R \mu^t_{j}$ for $j=0,\ldots,n$ by induction, which will lead to the proof of Theorem \ref{theorem:sim} following Lemma \ref{lemma:mu}.
	
	For $j=0$, $Pr(Path_t)\geq Pr(Path_s)$ holds trivially. For $j=1$, we have $\mu^s_1 \sqsubseteq_R \mu^t_1$ since $\mu^s_0$ and $\mu^t_0$ are generated directly from $\mu_{s_0}$ and $\mu_{t_0}$, respectively. 
	
	Assume $\mu^s_i \sqsubseteq_R \mu^t_i$ for $j=i$. Then for $j=i+1$,
	\begin{align*}
	\mu^s_{i+1}(s_y)&=\sum_{s_x\in S(i)}\mu^s_{i}(s_x)\cdot \mu_{s_x}(s_y),~ \forall s_y\in S(i+1),\\
	\mu^t_{i+1}(t_y)&=\sum_{t_x\in T(i)}\mu^t_{i}(t_x)\cdot \mu_{t_x}(t_y),~ \forall t_y\in T(i+1).
	\end{align*} 
	Since $\mu^s_i \sqsubseteq_R \mu^t_i$ , there exists a weight function $w$, such that	 
	\begin{align*}
	\mu^s_{i+1}(s_y)&=\sum_{s_x\in S(i)}\left[\sum_{t_1\in R(s_x)}w(s_x,t_1)\cdot \mu_{s_x}(s_y)\right],\\
	\mu^t_{i+1}(t_y)&\geq\sum_{t_x\in T(i)}\left[\sum_{s_1\in R(t_x)}w(s_1,t_x)\cdot \mu_{t_x}(t_y)\right].
	\end{align*} 	 
	Since $w(s,t)>0$ if and only if $s R t$, 
	
	\begin{align*}
	\mu^s_{i+1}(s_y)&=\sum_{s_x}\sum_{t_x}\left[w(s_x,t_x)\cdot \mu_{s_x}(s_y)\right],\\
	\mu^t_{i+1}(t_y)&\geq\sum_{s_x}\sum_{t_x}[w(s_x,t_x)\cdot \mu_{t_x}(t_y)]
	\end{align*} 
	with $s_x R t_x$. Because $s_x R t_x$, we have $\mu_{s_x}\sqsubseteq_R \mu_{t_x}$ with a weight function $w_x$. Thus 	 
	\begin{align*}
	\mu^s_{i+1}(s_y)&=\sum_{s_x}\sum_{t_x}\left[w(s_x,t_x)\cdot \sum_{t_2\in R(s_y)} w_x(s_y,t_2)\right]\\
	&=\sum_{t_2\in R(s_y)}\sum_{s_x}\sum_{t_x}w(s_x,t_x)\cdot w_x(s_y,t_2),\\
	\mu^t_{i+1}(t_y)&\geq\sum_{s_x}\sum_{t_x}\left[w(s_x,t_x)\cdot  \sum_{s_2\in R(t_y)} w_x(s_2,t_y)\right]\\
	&=\sum_{s_2\in R(t_y)}\sum_{s_x}\sum_{t_x}w(s_x,t_x)\cdot w_x(s_2,t_y).
	\end{align*} 	 
	Based on Definition \ref{def:mu_sim}, $\mu^s_{i+1}\sqsubseteq_R \mu^t_{i+1}$. By induction, we have proved $\mu^s_j\sqsubseteq_R \mu^t_{j}$ for $j=0,\ldots,n$. Thus we have $\mu^s_k(S(k))\leq\mu^t_k(T(k))$ following Lemma \ref{lemma:mu} given $T(k)=R(S(k))$. This also proves $Pr(Path_s) \leq Pr(Path_t)$. Because $n$ stands for an arbitrary length, by applying this proof to paths with different lengths in the counterexample path set from $\mathcal{M}^z_1$, we know that there always exists a finite set of paths in $\mathcal{M}^z_2$ with a larger accumulative probability that violates the specification. Then this corresponding set of paths is the witness of $\mathcal{M}^z_2$ violating $\phi$, which contradicts to our initial assumption. This concludes our proof for Theorem \ref{theorem:sim}. 	 

\end{proof}

With Theorem \ref{theorem:sim}, we have shown that safe simulation preserves finite horizon safe-PCTL specification $\phi$ between POMDP and the possible $0/1$-WA$^z$ as the abstract system. In next section, we present our CEGAR framework to find a proper $0/1$-WA$^z$ as the abstract system.

\section{Counterexample guided abstraction refinement}

In order to find a $0/1$-WA$^z$ $\mathcal{M}^z$ that can safely simulate the corresponding $0/1$-WA$^z$ of a POMDP $\mathcal{P}$ and capture enough properties from $\mathcal{P}$ to prove or disprove the satisfaction relation, we develop a novel CEGAR scheme. Since the model checking of specifications without probabilistic operators can be trivial, we focus on finite-horizon safe-PCTL with probabilistic operator $P$ and probability threshold.

\subsection{Quotient construction for $0/1$-WA$^z$}
\begin{figure}
	\centering
	\includegraphics[height=40mm]{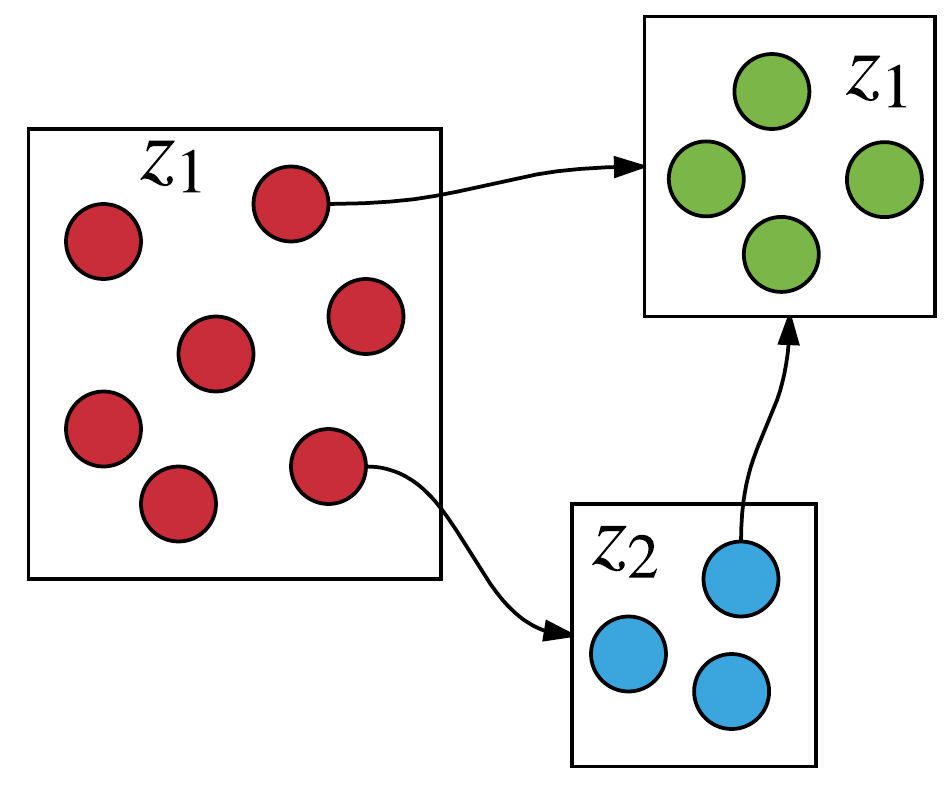}				
	\caption{Illustration of the quotient construction for  $0/1$-WA$^z$: the circles stand for the concrete states $s\in S$ and the rectangles stand for the abstract states generated by a given partition $\varPi$. Different colors stand for different atomic propositions. Based on Definition \ref{def:quotient}, the concrete states can belong to a same abstract state only if they have the same observation label.}				
\end{figure}

We first define the quotient $0/1$-WA$^z$ by partitioning the state space of the corresponding $0/1$-WA$^z$ $\mathcal{M}^z=\{S,\bar{s},A,T,Z,L,L^z\}$ generated for POMDP $\mathcal{P}$. This process is called quotient construction. Let $\varPi$ denote a partition of $S$ and $[s]_\varPi$ denote the equivalence class (abstract state) containing $s$ (concrete state). Here $\varPi$ can be viewed as a set of abstracted states and $[s]_\varPi\in\varPi$. Comparing with the quotient construction presented in the non-probabilistic case and MDP, we require $\forall s_1,s_2\in [s]_\varPi,L^z(s_1)=L^z(s_2),~L(s_1)=L(s_2)$. Such a partition $\varPi$ is called a consistent partition.
\begin{definition}\label{def:quotient}
	Given a $0/1$-WA$^z$ $\mathcal{M}^z=\{S,\bar{s},A,T,Z,L,L^z\}$ and a consistent partition $\varPi$, define the quotient $0/1$-WA$^z$ $\mathcal{M}^z/\varPi=\{\varPi,[\bar{s}]_\varPi,A,T_\varPi,Z,L^z_\varPi,L_\Pi\}$ where 
	\begin{align*}
	&T_\varPi([s]_\varPi,a,[s]'_\varPi)=\max_{s_i\in[s]_{\varPi}}\sum_{s_j\in[s]'_{\varPi}} T(s_i,a,s_j),\\ 
	&L^z_{\varPi}([s]_\varPi)=L^z(s),\\
	&L_{\varPi}([s]_\varPi)=L(s).
	\end{align*} 
\end{definition} 
\noindent Note that based on the requirement of the consistent partition, $[\bar{s}]_\varPi=\bar{s}$. By Definition \ref{def:safesim},  It follows straight forwardly that $\mathcal{M}^z/\varPi\preceq \mathcal{M}^z$ with safe simulation relation $R=\{(s,c)|s\in c,c\in\varPi\}$.

\subsection{CEGAR}
In order to find a quotient $0/1$-WA$^z$ $\mathcal{M}^z/\varPi$ that either shows 
$\mathcal{M}^z/\varPi\models\phi$ or provides the counterexample to show $\mathcal{M}^z\not\models\phi$, we follow a CEGAR approach. 

Initially, we construct the quotient $0/1$-WA$^z$ $\mathcal{M}^z/\varPi_0$ based on the coarsest partition $\varPi_0$ that groups any states with the same observation and atomic proposition labels together in $\mathcal{M}^z$. In each iteration $i$, if the model checking of the specification $\phi$ on $\mathcal{M}^z/\varPi_i$ returns yes, the CEGAR terminates with a proper quotient $0/1$-WA$^z$ $\mathcal{M}^z/\varPi_i$ as the abstract system that $\mathcal{M}^z/\varPi_i\models\phi$, which further implies $\mathcal{P}\models\phi$. Otherwise, we can get the counterexample as a finite set of finite paths $Path_{CE}$ and their accumulated probability mass $\ntrianglelefteq p$. Given the counterexample, we check whether this counterexample is spurious or not where the realizable probability of the counterexample is checked. If this counterexample is not spurious, the CEGAR terminates with the real counterexample that witnesses the violation of the given specification. Otherwise, we will use this spurious counterexample to refine the state space partition to get a finer state space partition. This process keeps going until the counterexample (CE) checking returns "No" in the sense that the previously found counterexample has been removed. After that, with the finer state space partition, we re-construct the quotient $0/1$-WA$^z$ for the iteration $i+1$ and do model checking again. The overview of our CEGAR approach is shown in Fig. \ref{fig:flowchart} and we illustrate our algorithms step-by-step as follows.

\begin{figure}
	\centering
	\includegraphics[height=75mm]{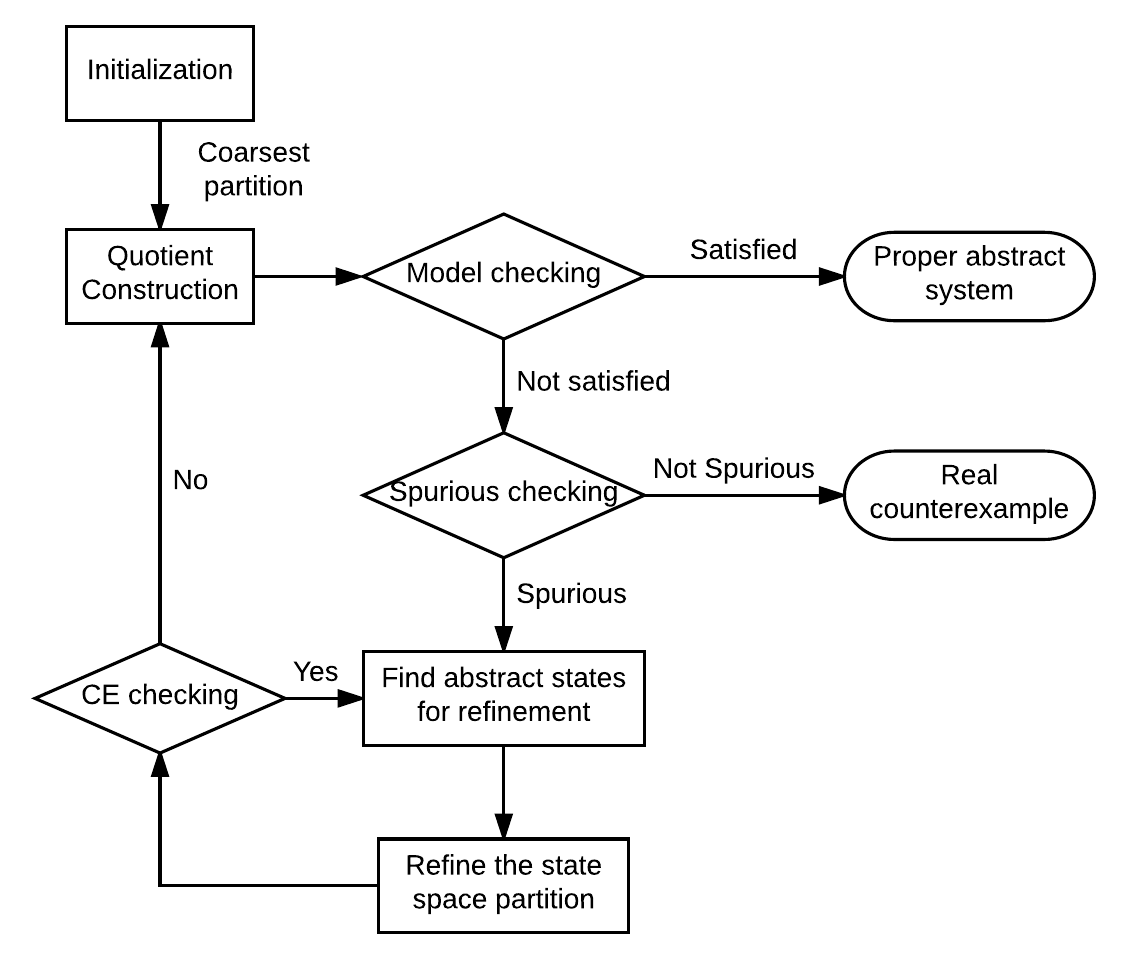}				
	\caption{The overview of the CEGAR approach for POMDP}		
	\label{fig:flowchart}	
\end{figure}

\subsubsection{Check spuriousness of counterexample}
Given the counterexample, we calculate the achievable probability in its realization in $\mathcal{M}^z$. For a abstract path $\rho_\varPi=c_0a_0...c_n\in Path_{CE}$, we define its realization in $\mathcal{M}^z$ as a set of concrete state paths denoted by $\pi(\rho_\varPi)=\{\rho|\rho=s_0a_0...s_n, s_j\in c_j, j=0,...,n\}$. If the accumulative probability mass of all realizations for $Path_{CE}$  $\ntrianglelefteq p$, we have found a real counterexample showing $\mathcal{P}\not\models\phi$. Otherwise, $Path_{CE}$ is spurious and introduced by a too coarse quotient partition, in which case we need to refine the state space partition and get a finer abstract system. The algorithm for spuriousness checking of a given counterexample is summarized in Algorithm \ref{alg:spurious}.

\begin{algorithm}
	\DontPrintSemicolon 
	\KwData{Counterexample as a set of paths $Path_{CE}$ and the probability threshold $p$ for the specification $\phi$}
	\KwResult{$true$ if the counterexample is spurious.}
	
	$sum\gets 0$\;
	\For{$\rho_\varPi=c_0a_0...c_n\in Path_{CE}$} {
		
		\For{$\rho\in \pi(\rho_\varPi)$}{$sum= sum + Pr^\rho(0,|\rho|)$\;
			\If{$sum\ntrianglelefteq p $}{\Return $false$}}
	}
	
	\Return $true$
	\caption{Check spuriousness of the counterexample}
	\label{alg:spurious}
\end{algorithm}

\subsubsection{Find the set of abstract states for refinement}

Given a spurious counterexample $Path_{CE}$ as a set of abstract state paths in $\mathcal{M}^z/\varPi$, the refinement algorithm first finds a set of abstract states needed to be refined then analyzes the splitting policy for certain abstract states to reduce the spuriousness of the counterexample. Based on the splitting policy, the refinement algorithm generates a new partition $\varPi$ for next iteration.

\begin{figure}
	\centering
	\includegraphics[height=40mm]{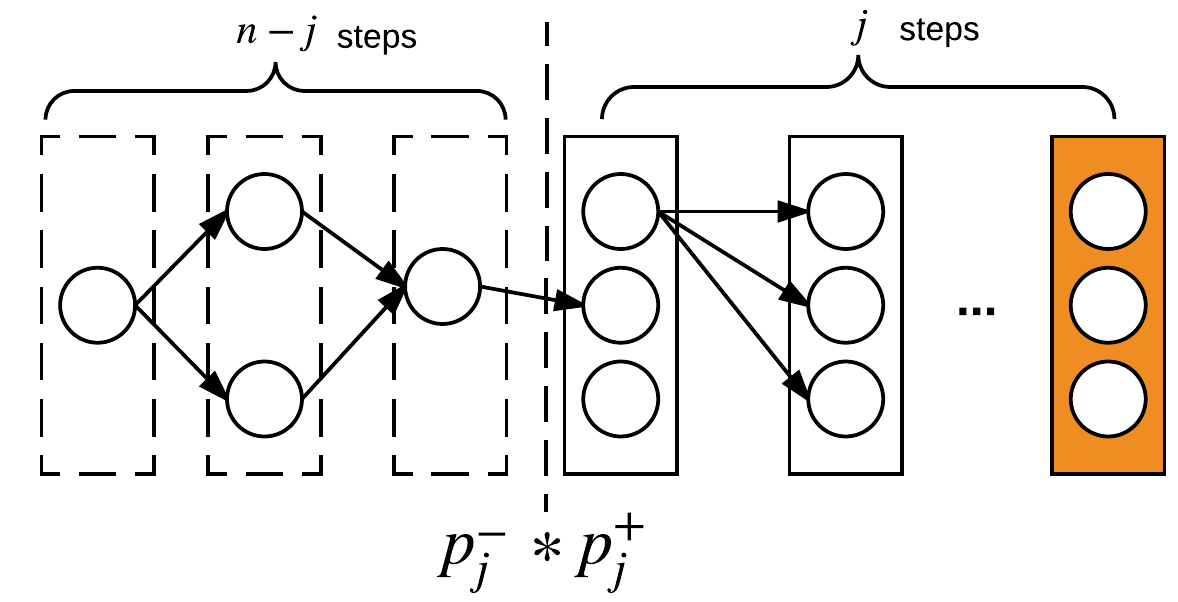}				
	\caption{Illustration of finding $B_j$. Starting from the end of a counterexample path, $j$ will separate this path into two parts. For the accumulative probability, we use the transition probability from the abstract state path to get $p^+_j$ for the last $j$ steps and the transition probability from the concrete state paths to get $p^-_j$ for the first $n-j$ steps.}		
	\label{fig:findBj}	
\end{figure}

To find a set of abstract states needed to be refined, we initialize $j=0$ and get a set of abstract states $B_j$ that contains the ($pivot$)th state for each $\rho_\varPi\in Path_{CE}$ with $pivot=\max(|\rho_\varPi|-j,0)$: 
\begin{align*}
B_j= \{c\in\varPi|c=\rho_\varPi(pivot),\rho_\varPi\in Path_{CE}\}.
\end{align*}
For each $\rho_\varPi\in Path_{CE}$, we calculate 
\begin{align*}
p^+_j= Pr^{\rho_\varPi}(pivot,|\rho_\varPi|).
\end{align*}
For each $\rho\in\pi(\rho_\varPi)$, we calculate 
\begin{align*}
p^-_j= Pr^{\rho}(0,pivot).
\end{align*}
Given $p^+_j$ and $p^-_j$ for each $\rho_\varPi\in Path_{CE}$ and $\rho\in\pi(\rho_\varPi)$, we can get the accumulated summation $SP_j=\sum_{\rho_\varPi\in Path_{CE}}\sum_{\rho\in\pi(\rho_\varPi)}p^-_j\cdot p^+_j$ and check whether $SP_j\unlhd p$ or not. If yes, we increase $j$ by 1 and go to the next iteration. If not, $B_j$ contains the abstract state needed to be refined. Notice that $SP_j$ is monotonic increasing and when $j$ reaches the maximum path length in $Path_{CE}$, $SP_j=\sum_{\rho_\varPi\in Path_{CE}}Pr^{\rho_{\varPi}}(0,|\rho_{\varPi}|)\ntrianglelefteq p$ because $Path_{CE}$ is a counterexample. Also because $Path_{CE}$ is spurious, $SP_1\trianglelefteq p$. These two facts guarantee the termination of the algorithm for some $j>0$. If $B_j$ is returned as a set of abstract states needed to refine, $B_j$ must have at least one abstract state that contains more than one concrete states; otherwise $sum$ will be the same for $j$ and $j-1$ which will make the algorithm end in iteration $j-1$. A direct implementation of the algorithm for this part is summarized in Algorithm \ref{alg:find}. 

\noindent\textbf{Remark}: While the calculation of $p^-_j$ requires the projection of an abstract state path to a set of concrete state paths, we only need to do the calculation once when $j=0$. In the initial iteration with $j=0$, we will calculate the probability of reaching each concrete state following the projection paths of the counterexample and these values can be saved to memory. In later iterations, the projection paths are still the same so we just need to load $p^-_j$ from memory. And since the counterexample paths have already solved the nondeterminism, solving $p^-_j$ can be done very efficiently by, for example, dynamic programming. 

\begin{algorithm}
	\DontPrintSemicolon 
	\KwData{spurious counterexample as a set of paths $Path_{CE}$ and the probability threshold $p$ for the specification $\phi$}
	\KwResult{a set of abstracted states $B_j$.}
	$j\gets -1$\;	
	
	\Do{$sum\unlhd p$}{
		$j\gets j+1$\;
		$sum\gets 0$\;
		
		$B_j\gets \{c\in\varPi|c=\rho_\varPi(\max(|\rho_\varPi|-j,0)),\rho_\varPi\in Path_{CE}\} $\;
		\For{$\rho_\varPi\in Path_{CE}$} {	
			$pivot\gets\max(|\rho_\varPi|-j,0)$\;	
			
			$p^+_j\gets Pr^{\rho_\varPi}(pivot,|\rho_\varPi|)$\;
			\For{$\rho\in \pi(\rho_\varPi)$}{
				$p^-_j\gets Pr^{\rho}(0,pivot) $\;
				$sum= sum+p^-_j\cdot p^+_j$\;	
				\If{$sum\ntrianglelefteq p $}{\Return $B_j$}}}

	} 
	\caption{Find the set of abstract states for refinement}
	\label{alg:find}
\end{algorithm}

\subsubsection{Refine the state space partition}

After finding $B_j$, we have a set of abstract states showing that the state space partition $\varPi_i$ is too coarse. We want to find a new partition to reduce the spuriousness of the counterexample. 

Assume the spurious counterexample $Path_{CE}$ contains only one abstract path $\rho_\varPi$ with the corresponding abstract state in $B_j$ having more than one concrete states. Let $c=\rho_\varPi(pivot)$ and $c_{next}=\rho_\varPi(pivot+1)$ with $pivot=\max(|\rho_\varPi|-j,0)$ and $a\in A$ the selected action for the transition between $c$ and $c_{next}$. We want to find the split policy that splits the equivalent class $c$ into two abstract states to reduce the value of $SP_j$. Based on the quotient construction rules for the abstract system, we could first get the array of transition probabilities from $s\in c$ to $c_{next}$ ($\sum_{s'\in c_{next}}T(s,a,s')$), then sort this array in descending order to get $G^+$. Here the $n$th element of $G^+$ stands for $\sum_{s'\in c_{next}}T(s[n],a,s')$, $n=1,2,...,|c|$. If we split $c$ into two parts that the first abstract state $c_1$ contains concrete states $\{s[1],...,s[n-1]\}$ and the second abstract state $c_2$ contains concrete states $\{s[n],...,s[|c|]\}$, then $T_\varPi(c_1,a,c_{next})=G^+[1]$ and $T_\varPi(c_2,a,c_{next})=G^+[n]$. Let $G^-$ represent an array of transition probabilities that its $n$th element is $G^-[n]=\sum_{\rho\in \pi(\rho_\varPi),\rho(pivot)=s[n]}Pr^\rho(0,pivot)$. Since we want to reduce the value of $SP_j$, we can go through the possible split policy by selecting the separation point $n$ ($n=2,3,...,|c|$) such that
\begin{align*}
\arg_{n}\min\sum^{n-1}_{d=1}G^-[d]\cdot G^+[1]+\sum^{|c|}_{d=n}G^-[d]\cdot G^+[n],
\end{align*}  
with the minimization goal describing the potential new $SP_j$. 

\begin{algorithm}
	\DontPrintSemicolon 
	\KwData{spurious counterexample as a set of paths $Path_{CE}$ in $\mathcal{M}^z/\varPi$ and a set of abstracted states $B_j$}
	\KwResult{new partition $\varPi$}	
	$Path \gets \{\rho_\varPi|\rho_\varPi\in Path_{CE},|\rho_\varPi(pivot)|>1 \text{~with~}pivot=\max(|\rho_\varPi|-j,0)\}$\;	
	$M= |Path|$\;
	\For{$m=1,...,M$}{	
		$\rho^m_\varPi\gets$ the $m$th path in $Path$\;
		$pivot_m=\max(|\rho^m_\varPi|-j,0)$\;
		$c^m=\rho^m_\varPi(pivot_m)$, $c^m_{next}=\rho^m_\varPi(pivot_m+1)$\;
		$a_m\gets$ the selected action between $c^m$ and $c^m_{next}$ on $\rho^m_\varPi$\;  
		generate $G^+_m$ and $G^-_m$\;
		$p^+_{j|\rho^m_\varPi}=Pr^{\rho^m_\varPi}(pivot_m,|\rho^m_\varPi|)$;\;
		\For{$n=2,...,|c^m|$}{
			
			$r_{mn}=\left(\sum^{n-1}_{d=1}G^{-}_m[d]\cdot G^+_m[1]+\sum^{|c|}_{d=n}G^-_m[d]\cdot G^+_m[n]\right)\cdot \frac{p^+_{j|\rho^m_\varPi}}{T_{\varPi}(c^m,a_m,c^m_{next})}$\;	
			
		}	
	}
	$m,n\gets \arg_{m,n}\min r_{mn}$\;
	update $\varPi$ by splitting $c_m$ with the separation point $n$\;
	\Return $\varPi$	
	\caption{Refine the state space partition}
	\label{alg:refine}
\end{algorithm}

If we have multiple abstract paths like $\rho_\varPi$, in this case we can go through each path, find the split policy based on that path, then compare and select the split policy that gives the best reduction for the potential new $SP_j$. Given a set of abstract state paths $Path_{CE}$ as the spurious counterexample and $B_j$, we first get a subset of abstract state paths $Path=\{\rho_\varPi|\rho_\varPi\in Path_{CE},|\rho_\varPi(pivot)|>1 \text{~with~}pivot=\max(|\rho_\varPi|-j,0)\}$. Basically we extract a set of abstract state paths with the corresponding abstract state in $B_j$ containing more than one concrete state. Let $|Path|=M$. For each $\rho^m_\varPi\in Path,~m=1,...,M$, we have $pivot_m$, $c^m$, $c^m_{next}$, $a_m$, $G^+_m$ and $G^-_m$, respectively. If the $m$th path is selected with separation point $n$, the potential new $SP_j$ can be captured by 
\begin{align*}
r_{mn}=\left(\sum^{n-1}_{d=1}G^{-}_m[d]\cdot G^+_m[1]+\sum^{|c|}_{d=n}G^-_m[d]\cdot G^+_m[n]\right)\\
\cdot \frac{p^+_{j|\rho^m_\varPi}}{T_{\varPi}(c^m,a_m,c^m_{next})},
\end{align*}
and the selected $m$ and $n$ should minimize $r_{mn}$. The algorithm for this part is summarized in Algorithm \ref{alg:refine}.

With the $m$ and $n$ minimizing the potential new $SP_j$, we have a split policy that splits the select abstract state into two. Under the new partition rule, we check the new counterexample paths. If these abstract state paths do not have enough accumulative probability of violation, we will move to the next iteration $i+1$ with the new state space partition $\varPi_{i+1}$; otherwise, we will apply Algorithm \ref{alg:find} to update $B_j$ then run a new round of refinement.

\begin{theorem}\label{theorem:cegar}
	The CEGAR framework is sound and complete. 
\end{theorem}
\begin{proof}
	Soundness: we initialize the coarsest assumption as a quotient $0/1$-WA$^z$ $\mathcal{M}^z/\varPi$. For every iteration $i\in \mathbf{N}$, the new generated abstract system $\mathcal{M}^z/\varPi_i$ is also a quotient $0/1$-WA$^z$. Follow the definition of quotient $0/1$-WA$^z$ and Theorem \ref{theorem:sim}, the soundness can be concluded. 
	
	Completeness: Since the newly generated abstract system is guaranteed to be finer than the older one in the previous iteration and the original concrete system $\mathcal{M}^z$ is the finest abstract system for itself, the convergence of the CEGAR is guaranteed. This shows in the worst case we find the original concrete system as its assumption, which concludes the completeness. 
	
\end{proof}

\section{Example}

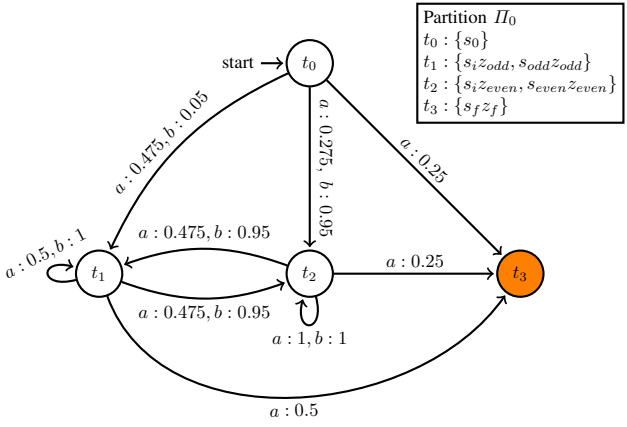
\begin{figure}
	\centering
	\begin{tikzpicture}[shorten >=1pt,node distance=4cm,on grid,auto, bend angle=20, thick,scale=0.7, every node/.style={transform shape}] 
	\node[state,initial] (t0) {$t_0$};
	\node[state] (t2) [below=of t0] {$t_2$};
	\node[state] (t1) [left=of t2] {$t_1$};
	\node[state, fill=orange] (t3) [right=of t2] {$t_3$};
	\node (q2) [draw,right=of t0,align=left] {	
		Partition $\varPi_0$\\		$t_0:\{s_0\}$\\$t_1:\{s_iz_{odd},s_{odd}z_{odd}\}$\\$t_2:\{s_iz_{even},s_{even}z_{even}\}$\\$t_3:\{s_fz_f\}$};  	
	\path[->]
	(t0) edge  [bend right] node [pos=0.55, sloped, above] {$a:0.475,b:0.05$} (t1) 			
	edge node [pos=0.5, sloped, above] {$a:0.275,~b:0.95$} (t2)
	edge node [pos=0.5, sloped, above] {$a:0.25$} (t3)
	(t1) edge [bend right] node [pos=0.5, sloped, below] {$a:0.475,b:0.95$} (t2)
	edge  [out=-70, in=235]node [pos=0.5, sloped, below] {$a:0.5$} (t3)
	edge [loop left, pos = 0.7,sloped, above] node {$a:0.5,b:1$} () 
	(t2) edge [bend right] node [pos=0.5, sloped, above] {$a:0.475,b:0.95$} (t1)
	edge node [pos=0.5, sloped, above] {$a:0.25$} (t3)
	edge [loop below, pos = 0.5,sloped, below] node {$a:1,b:1$} () 
	; 
	
	\end{tikzpicture} 
	\caption{The $0/1$-WA$^z$ $\mathcal{M}^z_0$ for $\mathcal{P}$ }
	\label{fig:M0}
\end{figure}

\begin{figure}
	\centering
	\begin{tikzpicture}[shorten >=1pt,node distance=4cm,on grid,auto, bend angle=20, thick,scale=0.7, every node/.style={transform shape}] 
	\node[state,initial, accepting] (t0) {};			  	
	\path[->]
	(t0)  edge [loop right,  align = center] node {$z_*,a$} () 
	; 
	\end{tikzpicture}
	\caption{The adversary that witnesses $\mathcal{M}_0\not\models\phi$ }
	\label{fig:adv}
\end{figure}
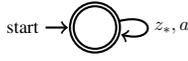

\begin{figure}
	\centering
	\begin{tikzpicture}[shorten >=1pt,node distance=4cm,on grid,auto, bend angle=20, thick,scale=0.7, every node/.style={transform shape}] 
	\node[state,initial] (t0) {$t_0$};
	\node[state] (t2) [below=of t0] {$t_2$};
	\node[state] (t1) [left=of t2] {$t_1$};
	\node[state] (t4) [below right=of t1] {$t_4$};
	\node[state, fill=orange] (t3) [right=of t2] {$t_3$};
	\node (q2) [draw,right=of t0,align=left] {	
		Partition $\varPi_1$\\		$t_0:\{s_0\}$\\$t_1:\{s_iz_{odd}\}$\\$t_2:\{s_iz_{even},s_{even}z_{even}\}$\\$t_3:\{s_fz_f\}$\\$t_4:\{s_{odd}z_{odd}\}$};  	
	\path[->]
	(t0) edge  [bend right] node [pos=0.55, sloped, above] {$a:0.475,b:0.05$} (t1) 			
	edge node [pos=0.5, sloped, above] {$a:0.275, b:0.95$} (t2)
	edge node [pos=0.5, sloped, above] {$a:0.25$} (t3)
	(t1) edge [bend right] node [pos=0.5, sloped, below] {$a:0.475,b:0.95$} (t2)
	edge [bend right] node [pos=0.5, sloped, below] {$a:0.25$} (t4)
	edge  [out=-55, in=235]node [pos=0.5, sloped, above] {$a:0.25$} (t3)
	edge [loop left, pos = 0.7,sloped, above] node {$a:0.475,b:0.95$} () 
	(t2) edge [bend right] node [pos=0.5, sloped, above] {$a:0.475,b:0.95$} (t1)
	edge node [pos=0.5, sloped, above] {$a:0.25$} (t3)
	edge [loop below, pos = 0.5,sloped, below] node {$a:1,b:1$} () 
	(t4) edge [out=-45, in=245] node [pos=0.5, sloped, below] {$a:0.5$} (t3)
	edge [loop below, pos = 0.5,sloped, below] node {$a:0.5,b:1$} () 
	; 
	
	\end{tikzpicture} 
	\caption{The $0/1$-WA$^z$ $\mathcal{M}^z_1$ }
	\label{fig:M1}
\end{figure}
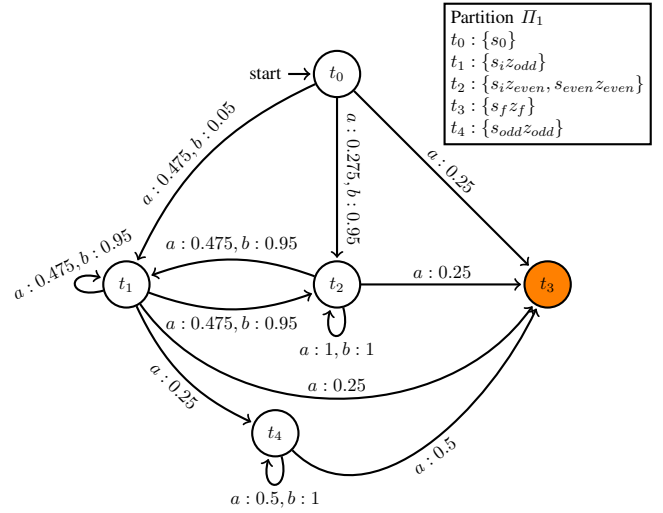	
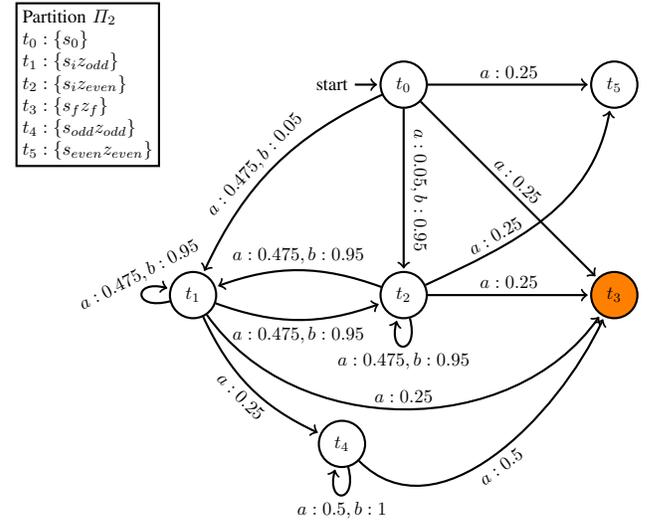
\begin{figure}
	\centering
	\begin{tikzpicture}[shorten >=1pt,node distance=4cm,on grid,auto, bend angle=20, thick,scale=0.7, every node/.style={transform shape}] 
	\node[state,initial] (t0) {$t_0$};
	\node[state] (t5) [right=of t0] {$t_5$};
	\node[state] (t2) [below=of t0] {$t_2$};
	\node[state] (t1) [left=of t2] {$t_1$};
	\node[state] (t4) [below right=of t1] {$t_4$};
	\node[state, fill=orange] (t3) [right=of t2] {$t_3$};
	\node (q2) [draw,left=of t0,xshift = -2cm, align=left] {	
		Partition $\varPi_2$\\		$t_0:\{s_0\}$\\$t_1:\{s_iz_{odd}\}$\\$t_2:\{s_iz_{even}\}$\\$t_3:\{s_fz_f\}$\\$t_4:\{s_{odd}z_{odd}\}$\\$t_5:\{s_{even}z_{even}\}$};  	
	\path[->]
	(t0) edge  [bend right] node [pos=0.55, sloped, above] {$a:0.475,b:0.05$} (t1) 			
	edge node [pos=0.5, sloped, above] {$a:0.05, b:0.95$} (t2)
	edge node [pos=0.5, sloped, above] {$a:0.25$} (t3)
	edge node [pos=0.5, sloped, above] {$a:0.25$} (t5)
	
	(t1) edge [bend right] node [pos=0.5, sloped, below] {$a:0.475,b:0.95$} (t2)
	edge [bend right] node [pos=0.5, sloped, below] {$a:0.25$} (t4)
	edge  [out=-55, in=235]node [pos=0.5, sloped, above] {$a:0.25$} (t3)
	edge [loop left, pos = 0.7,sloped, above] node {$a:0.475,b:0.95$} () 
	(t2) edge [bend right] node [pos=0.5, sloped, above] {$a:0.475,b:0.95$} (t1)
	edge node [pos=0.5, sloped, above] {$a:0.25$} (t3)
	edge[out=25, in=260] node [pos=0.3, sloped, above] {$a:0.25$} (t5)
	edge [loop below, pos = 0.5,sloped, below] node {$a:0.475,b:0.95$} () 
	(t4) edge [out=-45, in=245] node [pos=0.5, sloped, below] {$a:0.5$} (t3)
	edge [loop below, pos = 0.5,sloped, below] node {$a:0.5,b:1$} () 
	
	; 
	
	\end{tikzpicture} 
	\caption{The $0/1$-WA$^z$ $\mathcal{M}_2$}
	\label{fig:M2}
\end{figure}	

For POMDP system $\mathcal{P}$ described in Example \ref{example1}, consider the bounded until specification 	
$\phi=P_{\leq 0.45} (true~ \mathcal{U}^{\leq n} fail)$ and we want to check whether $\mathcal{P}\models \phi$. Let $n=20$. Based on the corresponding $0/1$-WA$^z$ $\mathcal{M}^z$ of $\mathcal{P}$, we get the initial coarsest abstraction for $\mathcal{P}$ by partitioning the state space of $\mathcal{M}^z$ into four equivalence classes $\{s_0\},\{s_iz_{odd},s_{odd}z_{odd}\},\{s_iz_{even},s_{even}z_{even}\}$ and $\{s_fz_f\}$ with $i\in[0,n)$. The resulting $0/1$-WA$^z$ $\mathcal{M}^z_0$ is shown in Fig. \ref{fig:M0} with four abstract states. Here $\mathcal{M}^z_0$ does not satisfy $\phi$ with the witness adversary shown in Fig. \ref{fig:adv}, which basically choose action $a$ for whatever history sequences. The counterexample is returned as a set of two paths: $\{t_0\rightarrow t_3, t_0\rightarrow t_1\rightarrow t_3\}$ with the accumulative probability of $0.4875$, which is larger than required $0.45$. It turns out that this CE is spurious. By checking the abstracted states backward in CE, we find $B_1=\{t_0,t_1\}$ that includes the over abstracted state and requires further refinement. The selected splitting policy splits $s_{odd}z_{odd}$ out of $t_1$. The generated abstract paths from the previous counterexample only have the accumulative probability of $0.36875$, which implies the counterexample paths have been removed and the newly generated abstract system $\mathcal{M}^z_1$ is shown in Fig. \ref{fig:M1}. 

By reusing the adversary found in the previous iteration, we can show $\mathcal{M}^z_1$ does not satisfy $\phi$ neither. The counterexample is returned as a set of four paths: $\{t_0\rightarrow t_3, t_0\rightarrow t_1\rightarrow t_3, t_0\rightarrow t_2\rightarrow t_3, t_0\rightarrow t_2\rightarrow t_2\rightarrow t_3\}$ with the accumulated probability of $0.50625$. This counterexample is also spurious. After finding $B_1=\{t_0,t_1,t_2\}$ that contains the over abstracted states, we choose the splitting policy that splits $s_{even}z_{even}$ out of $t_2$. And the generated abstract paths from the previous counterexample only have the accumulative probability around $0.3775$, which implies the counterexample paths have been removed and the newly generated abstract system $\mathcal{M}^z_2$ is shown in Fig. \ref{fig:M2}. 

Still with the adversary in Fig. \ref{fig:adv}, the dissatisfaction of $\mathcal{M}^z_2$ on $\phi$ is witnessed. This time the returned counterexample contains four paths: $\{t_0\rightarrow t_3, t_0\rightarrow t_1\rightarrow t_3, t_0\rightarrow t_1\rightarrow t_2\rightarrow t_3, t_0\rightarrow t_1\rightarrow t_4\rightarrow t_3\}$ with the accumulative probability around $0.4845$. It turns out that the realizable probability of the counterexample is larger than $0.45$, which means we have found the real counterexample showing that $\mathcal{P}\not\models\phi$.

\section{Conclusion and Future Work}
In this article, a CEGAR framework for POMDPs is proposed with the proof of the soundness and completeness. Inspired by strong simulation relation for probabilistic systems and CEGAR frameworks for MDPs, we define $0/1$-WA$^z$ as the form of the abstract system for POMDP and use safe simulation relation to preserving safe-PCTL with finite horizon. With $0/1$-WA$^z$ and safe simulation relation, the abstraction refinement algorithm is given to find a proper abstract system for POMDP based on counterexamples returned from model checking on the abstract systems iteratively. 

In the future, we will extend our results to other temporal logic specifications and relax the requirement on the specification to have a finite horizon. Meanwhile, we will completely implement the software package for our CEGAR framework.

\appendix\section*{Connection between model checking and optimal policy computation for $0/1$-WA$^z$}

%

In the optimal policy computation problem for POMDP $\mathcal{P}=\{S,\bar{s},A,Z,T,O\}$, a reward function $R(s,a):S\times A\rightarrow \mathbb{R}$ is defined to assign a numerical value quantifying the reward of performing an action $a$ at state $s$ \cite{pineau2006anytime}. The objective is to compute a policy (adversary) for selecting actions based on histories. The optimal policy $\sigma$ maximizes the expected future accumulative reward and usually the future rewards are discounted by a factor $\gamma: 0\leq \gamma <1$ to guarantee the accumulative reward is finite. Formally, given a history $h_t=a_0z_1,...,z_{t-1}a_{t-1}z_t$, the belief distribution $b(\cdot|h_t)$ over $S$ is the sufficient statistic \cite{astrom1965optimal}. Here $b(s|h_t)=Pr(s_t=s|h_t,b_0)$ stands for the conditional probability that the domain is in state $s$ at time $t$, given history $h_t$ and initial state distribution $b_0$. With the initial value function being defined as 
\begin{align*}
V_0(b|h)=\max_{a}\sum_{s\in S}R(s,a)b(s|h),
\end{align*}
the $t$th value function can be calculated from the ($t-1$)th following 
\begin{align*}
V_t(b|h)&=\max_{a}\left[\sum_{s\in S}R(s,a)b(s|h)\right.\\
&~~~~~~~~~~\left.+\gamma\sum_{z\in Z}Pr(z|a,b)V_{t-1}(\tau(b,a,z))\right],
\end{align*}
where $\tau()$ stands for the belief state update function following Bayesian rules \cite{pineau2006anytime}. For the bounded until safe-PCTL specification $\phi=P_{\unlhd p}[\phi_1~\mathcal{U}^{\leq k}\phi_2]$ considered in this article, one can convert the model checking problem to the optimal policy computation problem for POMDP by modifying the transition structure of POMDP to make states $s\models\neg\phi_1$ and states $s\models\phi_2$ absorbing, and designing the reward scheme that assigns 0 to intermediate transitions and 1 to the final transitions on $s\models\phi_2$ when depth $k$ is reached \cite{sharan2014formal,silver2010monte}. Since the future accumulative rewards are collected from the finite horizon, the discount factor is not needed. Under this reward scheme, the value function $V_k(b)$ stands for the accumulative probability of paths satisfying $\phi_1~\mathcal{U}^{\leq k}\phi_2$ with initial belief state $b$ under the optimal policy. Then the model checking problems can be answered by finding the optimal policy and checking whether $V_k(b_0)\unlhd p$ or not. This classic POMDP optimal policy computation problem that can be solved by, for example, value iteration method \cite{zhang2001speeding,pineau2006anytime}. Among different solvers for POMDP optimization problems, partially observable Monte-Carlo planning (POMCP) is a promising method based on Monte-Carlo simulation to find the optimal policy with convergence guarantee in large POMDPs \cite{silver2010monte,kocsis2006bandit}. 

In Section \ref{sec3}, we introduce $0/1$-WA$^z$ which is an extension of $0/1$-WA. With the requirement of the adversaries being observation-based, $0/1$-WA$^z$ can be seen as a special POMDP with the observation function for each state being Dirac delta function. For this special POMDP, the transition functions for some states satisfy $T(s,a,s')\in [0,1],~s,s'\in S,a\in A$ but $\sum_{s'\in S}T(s,a,s')$ can be larger than 1 while in the traditional POMDP $\sum_{s'\in S}T(s,a,s')$ must be equal (or less) to 1. We call this special kind of POMDP the weighted-POMDP. For weighted-POMDP, we can still design a reward scheme to connect its model checking problem with the optimal policy computation problem for the bounded until safe-PCTL specification $\phi=P_{\unlhd p}[\phi_1~\mathcal{U}^{\leq k}\phi_2]$. To illustrate this reward scheme, we notice that in the standard POMDP, the optimal policy computation problem can be understood as the planning problem on a derived MDP where the states of this MDP are the belief states in POMDP and the transition function is given by belief state updating function. Since the belief state corresponds to the history, the states of the derived MDP can also be understood as histories. Following this idea, we first present the construction of the derived MDP for weighted-POMDP then give the reward scheme for the derived MDP.   

For weighted-POMDP, the belief state distribution $b(\cdot|h_t)$ is no longer the sufficient statistic for a history $h_t$ because at each time instance the carried probability mass over $S$ is enlarged when $\sum_{s'\in S}T(s,a,s')>1$. To represent the total probability mass over $S$ for a history $h$, we define a coefficient function $C()$ that maps a history in weighted-POMDP to a positive real value. Formally, given initial state $\bar{s}$ we have $C(h_0)=1$ and 
\begin{align*}
b(s|h_0)=
\begin{cases}
1, s=\bar{s};\\
0,otherwise.
\end{cases}
\end{align*} 
The updates of the belief state $b$ and the coefficient function $C$ follow 
\begin{align*}
& b(s'|h_t)\\
&~~~~=\frac{\sum_{s\in S}O(s',z)T(s,a,s')b(s|h_{t-1})C(h_{t-1})}{\sum_{s''\in S}\sum_{s\in S} O(s'',z)T(s,a,s'')b(s|h_{t-1})C(h_{t-1})},\\
&C(h_t)=C(h_{t-1})\sum_{s\in S}\sum_{s'\in S}b(s|h_{t-1})T(s,a,s'), t=1,2,...
\end{align*}
Given a weighted-POMDP, its derived MDP $\mathcal{M}$ with histories as states is defined as $\mathcal{M}=\{H,A,T_M\}$ where $H$ is the state space of histories and 
\begin{align*}
& T_M(h,a,haz)\\~~~~&=\frac{\sum_{s''\in S}\sum_{s\in S} O(s'',z)T(s,a,s'')b(s|h)C(h)}{\sum_{z\in Z}\sum_{s''\in S}\sum_{s\in S} O(s'',z)T(s,a,s'')b(s|h)C(h)}\\
&~~~~=\frac{\sum_{s''\in S}\sum_{s\in S} O(s'',z)T(s,a,s'')b(s|h)C(h)}{C(haz)},
\end{align*}
where the second equality is established following the fact that $\sum_{z\in Z}O(s,z)=1$. Apparently $T_M(h,a,haz)$ is a standard transition function for MDP with \\$\sum_{z\in Z}T_M(h,a,haz)=1,\forall a\in A,~h\in H$.  

Given $\phi=P_{\unlhd p}[\phi_1~\mathcal{U}^{\leq k}\phi_2]$ as the specification for a weighted-POMDP, we still modify the transition structure to make states $s\models\neg\phi_1$ and states $s\models\phi_2$ absorbing. Instead of assigning reward for intermediate transitions, an immediate reward is assigned at the end for a history. Let $\Delta$ denote the set of $s\models \phi_2,s\in S$. Then the initial value function is defined as 
\begin{align*}
V_0(h)=C(h)\sum_{s\in\Delta}b(s|h)
\end{align*}
and the $t$th value function can be calculated from the ($t-1$)th:
\begin{align*}
V_t(h)=\max_{a}\sum_{z\in Z}T_M(h,a,haz)V_{t-1}(haz).
\end{align*}

\begin{theorem}
	$V_k(h)$ equals to the maximum accumulative probability of paths satisfying $\phi_1~\mathcal{U}^{\leq k}\phi_2$ with initial history $h$ following the optimal policy.
\end{theorem}
\begin{proof}
	We give an induction proof for this theorem. 
	
	For the base case of $k=0$, the theorem holds trivially with \begin{align*}
	V_0(h)=C(h)\sum_{s\in\Delta}b(s|h).
	\end{align*}
	
	Assume the theorem is true for $k=i,~i\geq 0$. Then we have 
	\begin{align*}
	V_i(h)= \sum_{s'\in S}b(s'|h)C(h)Pr\{s'z\xrightarrow{i}\Delta|h\},
	\end{align*}
	where $Pr\{s'z\xrightarrow{i}\Delta|h\}$ represents the maximum probability mass of reaching $\Delta$ in less or equal to $i$ steps from $s'$ conditional on history $h$ with $z$ as the latest observation.
	
	For $k=i+1$, the maximum accumulative path probability mass of reaching $\Delta$ in less or equal to $i+1$ steps for history $h$ can be written as 
	\begin{align*}
	Prob &= \max_a\sum_{z\in Z}\sum_{s'\in S}\sum_{s\in S} O(s',z)T(s,a,s')b(s|h)C(h) \\
	& ~~~~~~~~Pr\{s'z\xrightarrow{i}\Delta|haz\}\\
	&=\max_a\sum_{z\in Z} \sum_{s'\in S}b(s'|haz)T_M(h,a,haz)C(haz) \\& ~~~~~~~~Pr\{s'z\xrightarrow{i}\Delta|haz\}.
	\end{align*}
	
	Following the expression of $V_i(h)$, we can get
	\begin{align*}
	Prob=\max_a\sum_{z\in Z} T_M(h,a,haz)V_i(haz)=V_{i+1}(h).
	\end{align*}
	Thus $V_{i+1}(h)$ is the maximum probability mass of reaching $\Delta$ in less or equal to $i+1$ steps for the history $h$, which concludes the proof.
	\qed
\end{proof}

Then the model checking problems for the weighted-POMDP $\mathcal{P}$ on \\ $\phi=P_{\unlhd p}[\phi_1~\mathcal{U}^{\leq k}\phi_2]$ can be answered by finding the optimal policy that maximize $V_k(h_0)$ and checking whether $V_k(h_0)\unlhd p$ or not. This optimal policy computation problem can be solved using, for example, POMCP.  





\ifCLASSOPTIONcaptionsoff
  \newpage
\fi



\bibliographystyle{IEEEtran}
\bibliography{ref,ref1,ref3,ref2,Cyber-RelatedWork1}
%

%







\end{document}